\documentclass[journal,doublecolumn]{IEEEtran}

\usepackage{times}
\usepackage{graphicx}
\usepackage{color}
\usepackage[margin=1in]{geometry}
\usepackage[draft=false]{hyperref}

\usepackage{amssymb, amsthm, amsmath}

\newtheorem{theorem}{Theorem}
\newtheorem{lemma}{Lemma}
\newtheorem{Def}{Definition}
\newtheorem{Prop}{Proposition}
\newtheorem{Exam}{Example}
\newtheorem{Cor}{Corollary}
\newtheorem{rem}{Remark}
\newtheorem{fact}{Fact}

\def\<{\langle}
\def\>{\rangle}
\def\be{\begin{equation*}}
\def\ee{\end{equation*}}
\def\bea{\begin{eqnarray*}}
\def\eea{\end{eqnarray*}}

\newcommand{\comment}[1]{}

\def\ast{\dagger}

\newcommand{\tr}{\operatorname{Tr}}

\newcommand{\cE}{\mathcal{E}}
\newcommand{\cF}{\mathcal{F}}
\newcommand{\cI}{\mathcal{I}}
\newcommand{\cH}{\mathcal{H}}

\newcommand{\cD}{\mathcal{D}}
\newcommand{\cG}{\mathcal{G}}

\newcommand{\ox}{\otimes}

\newcommand{\ip}[2]{\langle #1 , #2\rangle}

\newcommand{\ket}[1]{\ensuremath{\left|#1\right\rangle}}
\newcommand{\op}[2]{|#1\rangle \langle #2|}
\newcommand{\Tr}{\mathop{\mathrm{tr}}\nolimits}

\newcommand{\setft}[1]{\mathrm{#1}}
\newcommand{\pos}[1]{\setft{Pos}\left(#1\right)}

\def\X{\mathcal{X}}
\def\Y{\mathcal{Y}}
\def\H{\mathcal{H}}
\def\Z{\mathcal{Z}}
\def\E{\mathcal{E}}
\newcommand{\lin}[1]{\setft{L}\left(#1\right)}
\begin{document}

\title {When is the Chernoff Exponent for Quantum Operations finite?}

\author{\IEEEauthorblockN{Nengkun Yu}\\
\IEEEauthorblockA{Centre for Quantum Software and Information,\\
   Faculty of Engineering and Information Technology, \\ University of
   Technology Sydney, NSW 2007, Australia\\
Email: nengkunyu@gmail.com}\\
 \IEEEauthorblockN{Li Zhou}\\
  \IEEEauthorblockA{Max Planck Institute for Security and Privacy, Bochum, Germany\\
  {zhou31416@gmail.com}}
}
 
\maketitle
\begin{abstract}
We consider the problem of testing two hypotheses of quantum operations in a setting of many
uses where an arbitrary prior probability distribution is given. The Chernoff exponent for quantum operations is investigated to track the minimal average error probability of discriminating two quantum operations asymptotically. We answer the question, ``When is the Chernoff exponent for quantum operations finite?''
We show that either two quantum operations can be perfectly distinguished with finite uses, or the minimal discrimination error decays exponentially with respect to the number of uses asymptotically. That is, the Chernoff exponent is finite if and only if the quantum operations can not be perfectly distinguished with finite uses. This rules out the possibility of super-exponential decay of error probability.
Upper bounds of the Chernoff exponent for quantum operations are provided. 
\end{abstract}

\section{Introduction}
%%%%%%%%%%%%%%%%%%%%%%%%%%%%%%%%%%%%%%%%%%%%%%%%%%%%%%%%%%%%%%%%%%%%%%
A fundamental problem in quantum information theory is to test a device that may be prepared
for implementing one of many quantum operations. The testing treated in the framework of quantum mechanics is performed by inputting a quantum state and performing a quantum measurement. The general noncommutative feature and the complex structure of quantum operations
make quantum statistics a much richer field than its classical counterpart.

In the degenerate case where the outputs of the quantum operations are fixed, the freedom of choosing input states becomes useless. The asymptotic behavior of the
average error, in discriminating a set of quantum states $\{\rho_1^{\ox n},\ldots,
\rho_r^{\ox n}\}$ with prior probability distribution $\{\Pi_1,\ldots,
\Pi_r\}$ is of great interest. In ~\cite{Parth01}, Parthasarathy
showed that the average error decays exponentially, asymptotically.
Significant efforts have been made to identify the Chernoff exponent as the optimal error exponent.
In two breakthrough papers, \cite{ACMMABV} and~\cite{NussbaumSzkola06}, the closed-form of the optimal error exponent was obtained, which can be regarded as
the quantum generalization of the Chernoff bound in classical
hypothesis testing~\cite{Chernoff52}. Li proved that multiple Chernoff exponent equals the minimal mutual Chernoff exponent \cite{Li2016}.

It is highly desirable to generalize the results of quantum states to quantum operations. Given that considerable experimental effort has been devoted to the field of quantum mechanics to prepare quantum systems and measure quantum states, it is of fundamental importance to develop a theory that can discriminate the different quantum operations. We note that for classical channel discrimination, the optimal exponential error rate problem has been well understood, where it is proven that adaptive choice does not improve the exponential error rate in these settings \cite{Hayashi09}.

Where quantum operations are only allowed to be used once, the problem has been extensively studied, with fruitful results.
By employing Holevo-Helstrom's celebrated theorem on the one-copy quantum state discrimination \cite{Helstrom1967,Holevo1972}, a completely bounded trace norm, known as the diamond norm, was introduced to characterize the difference between quantum channels by Kitaev \cite{Kitaev1997}. This norm becomes a fundamental
tool in almost all aspects of quantum information science \cite{Aharonov1997,Watrous2009,Watrous2015,Puzzuoli2016} since it is the most physically meaningful notion of distance between quantum operations.

The problem becomes much more complicated when quantum operations are used multiple times \cite{HHDW2010,CMW2016,TW2016}. 
Much effort has been devoted to characterizing the conditions of perfect distinguishability, in the sense that two quantum operations can be distinguished without error by a finite number of uses \cite{Acin2001,Duan2007,ZZG2007,LQ2008,Laing2009}.
Unlike classical channel discrimination, unitary operations exist that cannot be distinguished without error for single-use, while multiple uses can help achieve perfect discrimination. 
A complete solution to this problem is obtained in \cite{Duan2009} with a feasible, necessary, and sufficient condition.

In this paper, we investigate the concept of Chernoff exponent for quantum operations to characterize the asymptotic behavior of the average error probability of distinguishing given quantum operations under any prior probability distribution. Suppose we have a quantum device that is secretly chosen from $\{\cE, \cF\}$, and a known set of two quantum operations according to a prior probability distribution $\{\Pi_0,\Pi_1\}$. Our goal is to identify whether the device is $\cE$ or $\cF$ by using this device many times. We explore the Chernoff exponent for two quantum operations $\cE$ and $\cF$ to track the optimal error probability by using the following definition:
\begin{equation}\label{defineChernoffbound}
\xi_{\cE,\cF} = -\varlimsup_{n\rightarrow\infty} \frac{\log P_{err,min,n}}{n}
\end{equation}
where $P_{err,min,n}$ denotes the infimum discrimination error over all possible output states $\rho_n$ and $\sigma_n$ as illustrated in Figure 1 where the quantum device is used $n$ time.

\begin{figure}
\center
\includegraphics[width=8cm]{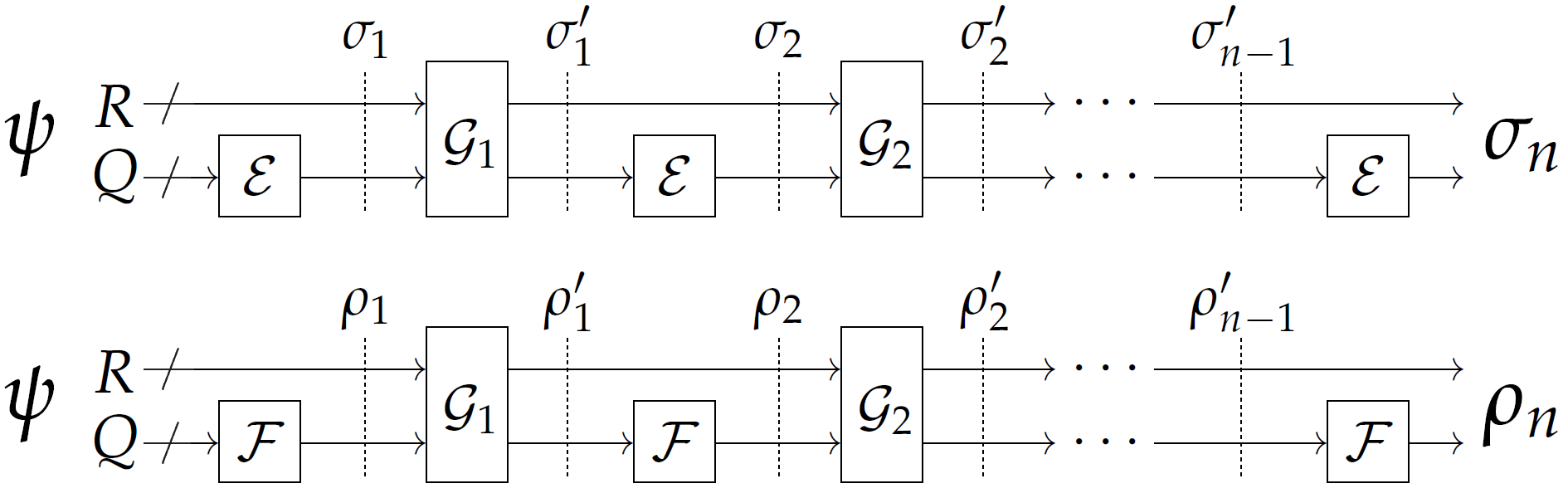}
\caption{Model of discriminating two quantum operations $\cE$ and $\cF$ on system $Q$ with $n$ uses by employing an arbitrary ancilla system $R$. $\cG_i$ are arbitrary quantum operations applied on the joint system $RQ$ and between two uses of the device. One can show that any arbitrary adaptive strategy can be translated into this model, by following the fact that any quantum measurement can be fully implemented by a quantum operation with outcomes stored as qubits.}  
\label{figure1}
\end{figure}

Notice that all possible strategies can be described by, or translate to, the model showed in Figure 1 with a sufficiently large ancilla system. 
This model is the most general scheme, and quantum operations $\mathcal{G}_i$ are freely chosen. For instance, parallel uses of devices can always be simulated by sequential uses and employing swap operators.

We show that the Chernoff exponent for quantum operations is finite if and only if they cannot be distinguished perfectly with finite uses. More precisely, we show that the average error probability decays at most, according to exponential function for quantum operations, if the quantum operations can not be perfectly distinguishable. This indicates that the error probability can never decay super-exponentially, such as $\exp(-\alpha n^2)$. Computable upper bounds on the Chernoff exponent for quantum operations are provided. Finally, we generalize our results to deal with multiple quantum operations.

\section{Notations and Preliminaries}

We use the symbols $\cH,\X,\Y,\Z$ to denote finite-dimensional Hilbert spaces over complex numbers and
$\lin{\cH}$ to denote the set of linear operators mapping from $\cH$ into itself.
For Hermitian matrices $A,B$, we use $\ip{A}{B}=\tr(A^{\dag}B)=\tr(AB)$ to denote their inner product.
Let $\mathrm{Pos}(\cH)\subset\lin{\cH}$ be the set of positive (semidefinite) matrices, and $\cD(\cH) \subset\mathrm{Pos}(\cH)$ is the set of positive matrices with trace one.
A pure quantum state of $\cH$ is just a normalized vector $\ket{\psi}\in\cH$, while a general quantum state is characterized by a density operator $\rho\in\cD(\cH)$.
For simplicity, we use $\psi$ to represent the density operator of a pure state $\ket{\psi}$ which is just the projector $\psi = \op{\psi}{\psi}$.
A density operator $\rho$ can always be decomposed into a convex combination of pure states:
\be
\rho = \sum_{k=1}^{n}p_k\op{\psi_k}{\psi_k},
\ee
where the coefficients $p_k$ are strictly positive numbers and add up to one. The support of $\rho$ is defined as $\mathrm{supp}(\rho) = \mathrm{span}\{\ket{\psi_k}: 1\leq k\leq n\}$.
We say two pure states $\ket{\psi}$ and $\ket{\phi}$ are orthogonal if and only if their inner product $\ip{\psi}{\phi}$ is equal to zero, and the orthogonality of two density operators $\rho$ and $\sigma$ is defined by the orthogonality of their supports, namely, $\rho$ and $\sigma$ are orthogonal if and only if $\mathrm{supp}(\rho)\perp\mathrm{supp}(\sigma)$.
Two density operators $\rho$ and $\sigma$ are said to be disjoint if $\mathrm{supp}(\rho)\cap\mathrm{supp}(\sigma) = \{0\}$ and joint if the intersection of their support contains some non-zero vectors.

There are two commonly used measures to characterize the difference between the quantum states: trace distance and fidelity. The trace distance $D$ between two density operators $\rho$ and $\sigma$ is defined as
\be
D(\rho,\sigma) \equiv \frac{1}{2}\mathrm{Tr}|\rho-\sigma|
\ee
where we define $|A|\equiv\sqrt{A^\dag A}$ to be the positive square root of $A^\dag A$.

The fidelity of states $\rho$ and $\sigma$ is defined to be
\be
F(\rho,\sigma) \equiv \mathrm{Tr}\sqrt{\sqrt{\rho}\sigma\sqrt{\rho}}.
\ee
For pure states $\ket{\psi}$ and $\ket{\phi}$, $F(\psi,\phi)=|\langle\psi|\phi\rangle|$.

The strong concavity property for the fidelity is quite useful, which can be formalized as
\begin{fact}[\cite{NielsenC00}]\label{strong concavity}
For quantum states $\rho_i$, $\sigma_i$ and probability distributions $(p_0,p_1,\cdots,p_n)$ and $(q_0,q_1,\cdots,q_n)$
\be
F\left( \sum_ip_i\rho_i,\sum_iq_i\sigma_i \right) \ge \sum_{i=0}^n\sqrt{p_iq_i}F(\rho_i,\sigma_i).
\ee
\end{fact}
If $\rho_i=\psi_i$ and $\sigma_i=\phi_i$ are all pure states, we obtain

\begin{align*}
F\left( \sum_ip_i\psi_i,\sum_iq_i\phi_i \right) \ge \sum_{i=0}^n\sqrt{p_iq_i}F(\psi_i,\phi_i)\\
=\sum_{i=0}^n|\langle\sqrt{p_i}\psi_i|\sqrt{q_i}\phi_i\rangle|.
\end{align*}

\begin{Def}\label{def:purification}
	We say that a pure state $\ket{\psi} \in \H_A\otimes\H_B$
	is a purification of some state $\rho$
	if $\Tr_A({\psi})=\rho$.
\end{Def}

\begin{fact}[Uhlmann's theorem, \cite{Uhl77}]
	\label{fac:Uhlmann}
	Given quantum states $\rho$, $\sigma$, and a purification $\ket{\psi}$ of $\rho$,
	it holds that $F({\rho},{\sigma})=\max_{\ket{\phi}} | \langle \phi | \psi \rangle | $,
	where the maximum is ranging over all purifications of $\sigma$.
\end{fact}

The following fact connects the trace distance and the fidelity between
two states.
\begin{fact}[Fuchs-van de Graaf inequalities~\cite{FuchsG99}]
	\label{fac:tracefidelityequi}
	For quantum states $\rho$ and $\sigma$, it holds that
\be
1-F(\rho,\sigma) \le D(\rho,\sigma) \le \sqrt{1-F(\rho,\sigma)^2}.
\ee
	For pure states $\ket{\phi}$ and $\ket{\psi}$, we have
	\begin{align*}
	D({\phi}, {\psi})= \sqrt{1 - F({{\phi}},{{\psi}})^2}=\sqrt{1 - |\langle\phi|\psi\rangle|^2}.
	\end{align*}
\end{fact}

The trace distance is a static measure quantifying how close two quantum states are and is closely related to the discrimination of quantum states.
Let us consider the two hypotheses, $H_0$ and $H_1$. Hypothesis $H_0$ assumes that a given unknown quantum state is equal to $\rho_0$, and Hypothesis $H_1$ assumes that a given unknown quantum state is equal to $\rho_1$. We assume that the prior probability distribution of $\rho_0$ and $\rho_1$ are $\Pi_0$ and $\Pi_1$, respectively, which add up to one.

A physical strategy to discriminate between these two hypotheses is to perform a positive-operator valued measure (POVM) on the quantum state with two outcomes, 0 and 1.
Such a POVM has two elements $\{E_0,E_1\}$ satisfying $E_0,E_1\in \mathrm{Pos}(\cH)$ and $E_0+E_1=I$, where $I$ is the identity matrix of $\cH$.
The aim of quantum state discrimination is to find the elements $E_0$ and $E_1$ that minimize the total error $P_{err}$, which is
\be
P_{err} = \Pi_0\mathrm{Tr}[E_1\rho_0]+\Pi_1\mathrm{Tr}[E_0\rho_1].
\ee

This optimal error has been identified by Helstrom as expressed in
the following equation
\be
P_{err,min} = \frac{1}{2}\left(1-\mathrm{Tr}|\Pi_1\rho_1-\Pi_0\rho_0|\right).
\ee

A quantum operation $\cE$ from $\lin{\cH}$ to $\lin{\Z}$ is a completely positive and trace-preserving map used to describe the evolution of an open quantum system.
A quantum operation $\cE$ can always be represented using the Kraus representation as
\be
\cE(\rho) = \sum_{i=1}^kE_i\rho E_i^\dag,
\ee
where $\{E_i\}_{i=1,\cdots,k}$ are the Kraus operators of $\cE$ satisfying $\sum_{i=1}^kE_i^\dag E_i=I$, the identity of $\cH$.

The following fact states that the fidelity between
two states is non-decreasing under quantum operations.
\begin{fact}[\cite{NielsenC00}]
	\label{fac:monotonequantumoperation}
	For states $\rho$, $\sigma$, and quantum operation $\E(\cdot)$, it holds that
	\begin{align*}
	F({\E(\rho)},{\E(\sigma)})&\geq F({\rho},{\sigma}).
	\end{align*}
\end{fact}

Quantum operations $\cE$ and $\cF$ are said to be perfectly distinguishable with finite uses if there exists a strategy illustrated as Figure \ref{figure1} such that $\sigma_n$ and $\rho_n$ are orthogonal.

Two conditions introduced by \cite{Duan2009} characterize the perfect distinguishability between quantum operations.
\begin{Def}
Two quantum operations, $\cE$ and $\cF$, acting on the same principal system, denoted by $Q$, are said to be disjoint if there is an auxiliary system $R$, and a pure state $\ket{\psi^{RQ}}$, such that $(\cI^R\otimes\cE^Q)(\psi^{RQ})$ and $(\cI^R\otimes\cF^Q)(\psi^{RQ})$ are disjoint, where
$\cI^R$ is the identity operation on $R$, and the superscripts only identify which systems the operations acted on. Otherwise, they are called joint.
\end{Def}
Intuitively, this disjointness guarantees that the outputs do not have a common part with the carefully chosen input. 
This disjointness is necessary to achieve perfect distinguishability. Otherwise, according to an inductive argument, there is always a non-zero common part between the outputs for an arbitrary strategy with finite uses.

Another relationship is to ensure that non-orthogonal states, $\rho$ and $\sigma$, exist such that $(\cI^R\otimes\cE^Q)(\rho^{RQ})$ and $(\cI^R\otimes\cF^Q)(\sigma^{RQ})$ become orthogonal, then one can distinguish $(\cI^R\otimes\cE^Q)(\rho^{RQ})$ and $(\cI^R\otimes\cF^Q)(\sigma^{RQ})$  without error. This is the final step of any strategy to achieve perfect distinguishability between $\cE$ and $\cF$.

Interestingly, these two conditions are not only necessary but also sufficient for the perfect distinguishability between quantum operations \cite{Duan2009} .
\begin{Prop}
Two quantum operations $\cE$ and $\cF$ are perfectly distinguishable if and only if: 1). They are disjoint; 2). They can map some non-orthogonal states into orthogonal states.
\end{Prop}
We remark here that ancillary systems are also allowed to achieve perfect discrimination.

In the following, we give an analytical characterization of the negation of the second condition, i.e., that  $\cE$ and $\cF$ cannot map some non-orthogonal states into orthogonal states even with the help of ancillary system. 
\begin{rem}
For $\cE(\cdot) = \sum_{i} E_i\cdot E_i^\dag$ and $\cF(\cdot) = \sum_{j} F_j\cdot F_j^\dag$, Condition 2) of Proposition 1 is equivalent to $I\notin\mathrm{span}\{( {E_i}^{\dag}F_j); 1\leq i,j\leq m\} $.
\end{rem}

We want to emphasize this proof of the characterization is precisely the same as given in \cite{Duan2009}. We provide the following argument for the readers' convenience.

Without loss of generality, we assume that $\cE$ and $\cF$ have the same number of Kraus operators by adding zero Kraus operators, if necessary. That is,
$\cE(\cdot) = \sum_{i=1}^m E_i\cdot E_i^\dag$ and $\cF(\cdot) = \sum_{j=1}^m F_j\cdot F_j^\dag$, cannot make non-orthogonal states orthogonal. That is, if $\rho^{RQ}$ and $\sigma^{RQ}$ are not orthogonal, then $(\cI^R\otimes\cE^Q)(\rho^{RQ})$ and $(\cI^R\otimes\cF^Q)(\sigma^{RQ})$ are not orthogonal. Equivalently, if pure states $\rho^{RQ}=\op{\psi^{RQ}}{\psi^{RQ}}$ and $\sigma^{RQ}=\op{\phi^{RQ}}{\phi^{RQ}}$ are not orthogonal, then $(\cI^R\otimes\cE^Q)(\rho^{RQ})$ and $(\cI^R\otimes\cF^Q)(\sigma^{RQ})$ are not orthogonal.
In other words, 
\begin{align*}
\tr[(\cI^R\otimes\cE^Q)(\psi^{RQ})(\cI^R\otimes\cF^Q)(\phi^{RQ})]=0 
\end{align*} 
implies
\[
 \langle\psi^{RQ}|\phi^{RQ}\rangle=0.
\]
That is, if $\forall 1\leq i,j\leq m$,
\begin{align*} 
(I^R\otimes E_i^Q) \op{\psi^{RQ}}{\psi^{RQ}}(I^R\otimes E_i^Q)^{\dag}
\end{align*} 
is orthgonal to 
\begin{align*} 
(I^R\otimes F_j^Q)\op{\phi^{RQ}}{\phi^{RQ}} (I^R\otimes F_j^Q)^{\dag}],
\end{align*} 
then
\begin{align*}
\langle\psi^{RQ}|\phi^{RQ}\rangle=0.
\end{align*} 
The above condition is equivalent to for $\ket{\psi^{RQ}}$ and $\ket{\phi^{RQ}}$, if
 \begin{align*}
\langle\psi^{RQ}|(I^R\otimes E_i^Q)^{\dag}(I^R\otimes F_j^Q)\ket{\phi^{RQ}}=0
\end{align*} 
for all $ 1\leq i,j\leq m$, then
 \[
 \langle\psi^{RQ}|\phi^{RQ}\rangle=0.\]
That is, if $\forall 1\leq i,j\leq m$
 \begin{align*}
 \langle\psi^{RQ}|(I^R\otimes {E_i}^{\dag}F_j) \ket{\phi^{RQ}}=0,
 \end{align*} 
 then
\[
 \langle\psi^{RQ}|\phi^{RQ}\rangle=0.
 \]
 For any $M\in \lin{\cH_Q}$, one can find $\ket{\phi^{RQ}}$ and $\ket{\psi^{RQ}}$ such that $M=\tr_R \op{\phi^{RQ}}{\psi^{RQ}}$. 
 We know that $\forall 1\leq i,j\leq m$
  \begin{align*}
  \tr (M{E_i}^{\dag}F_j)=\langle\psi^{RQ}|(I^R\otimes {E_i}^{\dag}F_j) \ket{\phi^{RQ}}=0,
 \end{align*} 
 implies
 \[
     \tr M=0.
    \]
That is satisfied if and only if $I^{RQ}\in\mathrm{span}\{(I^R\otimes {E_i}^{\dag}F_j); 1\leq i,j\leq m\} $, which in turn is equivalent to $I\in\mathrm{span}\{E_i^{\dag}F_j; 1\le i,j\le m\}$. 

Therefore, $\cE$ and $\cF$ can map some non-orthogonal states into orthogonal states, Condition 2) of Proposition 1, is equivalent to $I\notin\mathrm{span}\{( {E_i}^{\dag}F_j); 1\leq i,j\leq m\} $.
 
\section{Two useful Lemmas}

The following lemma shows that if two quantum operations are joint, then for any input state, the output states always have a common semi-definite positive component whose size is positive and depends only on the operations.
\begin{lemma} \label{jointbound}
If $\mathcal{E}$ and $\mathcal{F}$ are joint, there exists $\eta>0$, depending only on $\mathcal{E}$ and $\mathcal{F}$, such that for any quantum state $\rho$ on a
potentially larger Hilbert space $RQ$, there is a matrix $A$, such that $0\leq A\leq (\cI^R\otimes\mathcal{E})(\rho^{RQ}),(\cI^R\otimes\mathcal{F})(\rho^{RQ})$ and $\tr(A)\geq \eta$.
\end{lemma}
\begin{proof}
It is
straightforward to verify that we only need to consider $\rho$ to be a pure state. Thus, according to Schmidt decomposition, we can assume that $\rho$ is a quantum state in $\mathcal{H}_{RQ}=\mathcal{H'}\otimes\mathcal{H}$ with the dimension of $\mathcal{H}'$ being equal to the dimension of $\mathcal{H}$, where $\mathcal{H}$ is the Hilbert space of $Q$.

Our goal is to show
$$
\eta>0
$$
where $\eta$ is defined as
\begin{align*}
\eta:=\inf_{\rho\in \cD(\mathcal{H}_{RQ})}\sup_{\substack{0\leq X\leq (\cI^R\otimes\mathcal{E})(\rho^{RQ}),\\0\leq X\leq(\cI^R\otimes\mathcal{F})(\rho^{RQ})}}  \tr X \\=\inf_{\rho\in \cD(\mathcal{H}_{RQ})}\max_{\substack{0\leq X\leq (\cI^R\otimes\mathcal{E})(\rho^{RQ}),\\0\leq X\leq(\cI^R\otimes\mathcal{F})(\rho^{RQ})}}  \tr X.
\end{align*}
To prove this, we notice that for a fixed input $\rho^{RQ}$, the optimization problem
\[
\max_{\substack{0\leq X\leq (\cI^R\otimes\mathcal{E})(\rho^{RQ}),\\0\leq X\leq(\cI^R\otimes\mathcal{F})(\rho^{RQ})}}  \tr X
\]
can be formulated as the following semidefinite program \cite{Watrous2009}:
\begin{center}
  \begin{minipage}{2in}
    \centerline{\underline{Primal problem}}\vspace{-2mm}
    \begin{align*}
      \text{maximize:}\quad & \ip{I}{X}\\
      \text{subject to:}\quad & \Phi(X) \leq B,\\
      & X\in\pos{\mathcal{H}_{RQ}}.
    \end{align*}
  \end{minipage}
  \hspace*{12mm}
  \begin{minipage}{2in}
    \centerline{\underline{Dual problem}}\vspace{-2mm}
    \begin{align*}
      \text{minimize:}\quad & \ip{B}{Y}\\
      \text{subject to:}\quad & \Phi^{\ast}(Y) \geq I,\\
      & Y\in\pos{\mathcal{H}_{RQ}\oplus \mathcal{H}_{RQ}}.
    \end{align*}
  \end{minipage}
\end{center}
In the above formula,
$\Phi$ is the super-operator
\[
\Phi: \lin{\mathcal{H}_{RQ}} \rightarrow \lin{\mathcal{H}_{RQ}\oplus \mathcal{H}_{RQ}}
\]
as
\[
\Phi(X)
= \begin{pmatrix}X & 0\\ 0 & X\end{pmatrix},
\]
where the adjoint super-operator
\[
\Phi^{\ast}:\lin{\mathcal{H}_{RQ}\oplus \mathcal{H}_{RQ}}\rightarrow\lin{\mathcal{H}_{RQ}}
\]
is given by
\[
\Phi^{\ast}\begin{pmatrix} Z & \cdot \\ \cdot & W\end{pmatrix}
=Z+W,
\]
and
\[
B=\begin{pmatrix} (\cI^R\otimes\mathcal{E})(\rho^{RQ}) & 0 \\ 0 & (\cI^R\otimes\mathcal{F})(\rho^{RQ})\end{pmatrix}.
\]
Choose $Y=I>0$, then $\Phi^{\ast}(Y)=2I>0$. This dual program is strictly feasible. Thus, the primal value and dual value are the same \cite{Slater1950}.

Now we are going back to the original problem by considering the dual problem with $\rho^{RQ}$ ranging over all possible states. 

Let $\mathcal{B}$ denote the following set
\[
\bigg \{\begin{pmatrix} (\cI^R\otimes\mathcal{E})(\rho^{RQ}) & 0 \\ 0 & (\cI^R\otimes\mathcal{F})(\rho^{RQ})\end{pmatrix}:\rho \in \mathrm{D}(\mathcal{H}_{RQ})\bigg\}.
\]
$\mathcal{B}$ is a compact set because it is a closed bounded set in a finite-dimensional space.

According to the compactness of $\mathcal{B}$, we have the following
\begin{align*}
&\inf_{\rho\in\mathrm{D}(\mathcal{H}_{RQ})}\max_{\substack{0\leq X\leq (\cI^R\otimes\mathcal{E})(\rho),\\0\leq X\leq(\cI^R\otimes\mathcal{F})(\rho)}}  \tr X\\
=&\inf_{B\in \mathcal{B}}\min_{\substack{{\Phi^{\ast}(Y) \geq I,}\\{Y\in\pos{\mathcal{H}_{RQ}\oplus \mathcal{H}_{RQ}}}}}  \ip{B}{Y}\\
=&\min_{\substack{{\Phi^{\ast}(Y) \geq I,}\\{Y\in\pos{\mathcal{H}_{RQ}\oplus \mathcal{H}_{RQ}}}}}\inf_{B\in \mathcal{B}}\ip{B}{Y}\\
=&\min_{\substack{{\Phi^{\ast}(Y) \geq I,}\\{Y\in\pos{\mathcal{H}_{RQ}\oplus \mathcal{H}_{RQ}}}}}\min_{B\in \mathcal{B}}\ip{B}{Y}\\
%=&\min_{\substack{{\Phi^{\ast}(Y) \geq I,}\\{Y\in\pos{RQ\oplus RQ}},\\{B\in \mathcal{B}}}}\ip{B}{Y}\\
=&\ip{B_0}{Y_0},
\end{align*}
for some $B_0\in \mathcal{B}$ and $Y_0\in \pos{\mathcal{H}_{RQ}\oplus \mathcal{H}_{RQ}}$.

Let $\rho_0^{RQ}$ be such that
\begin{align*}
  Y_0&=\begin{pmatrix} M & \cdot \\ \cdot & N\end{pmatrix}\\
B_0&=\begin{pmatrix} (\cI^R\otimes\mathcal{E})(\rho_0^{RQ}) & 0 \\ 0 & (\cI^R\otimes\mathcal{E})(\rho_0^{RQ})\end{pmatrix}.
\end{align*}
For $\rho_0^{RQ}$, the intersection of the supports of $(\cI^R\otimes\mathcal{E})(\rho_0^{RQ})$ and $(\cI^R\otimes\mathcal{F})(\rho_0^{RQ})$ has a non-zero element. It indicates that there exists a non-zero $0\leq G\leq (\cI^R\otimes\mathcal{E})(\rho_0^{RQ}),(\cI^R\otimes\mathcal{F})(\rho_0^{RQ})$.

According to $\Phi^{\ast}(Y_0) \geq I$, we have $M+N\geq I$. Let
\begin{align*}
\eta=&\ip{B_0}{Y_0}\\
=&\ip{(\cI^R\otimes\mathcal{E})(\rho_0^{RQ})}{M}+\ip{(\cI^R\otimes\mathcal{F})(\rho_0^{RQ})}{N}\\
\geq& \ip{G}{M}+\ip{G}{N}\\
=& \ip{G}{M+N}\\
\geq & \ip{G}{I}\\
=& \tr G>0.
\end{align*}
We can conclude that this $\eta$ satisfies the wanted property.
\end{proof}

The following observation shows that if $I\in \mathrm{span}\{E_i^{\dag}F_j\}$, then $\mathcal{E}$ and $\mathcal{F}$ cannot change the fidelity of two quantum states significantly.
 \begin{lemma} \label{fidelity}
If $I\in \mathrm{span}\{E_i^{\dag}F_j\}$, then, there exists $\zeta>0$, depending only on $\mathcal{E}$ and $\mathcal{F}$, such that for all $\rho, \sigma$ on a
potentially larger Hilbert space $RQ$,
$$F((\cI^R\otimes\mathcal{E})(\rho^{RQ}),(\cI^R\otimes\mathcal{F})(\sigma^{RQ}))\geq \zeta F(\rho,\sigma).$$
\end{lemma}
\begin{proof}
The condition $I\in \mathrm{span}\{E_i^{\dag}F_j\}$ leads us to the existence of $\chi_{i,j}\in \mathbb{C}$ such that
\be
I=\sum_{i,j=1}^m \chi_{i,j} E_i^{\dag}F_j.
\ee
Using polar decomposition of the coefficient matrix $\chi_{i,j}$, we can always assume that
\be
I=\sum_{i=1}^m \chi_{i} E_i^{\dag}F_i,
\ee
and $\chi_i\geq 0$. We use $\chi=\max_{i} \chi_i$ to denote the largest $\chi_i$.

For any $\rho^{RQ}$ and $\sigma^{RQ}$, by Uhlmann's Theorem \ref{fac:Uhlmann}, there exist $\ket{\psi^{RQT}_{\rho}}$ and $\ket{\psi^{RQT}_{\sigma}}$ being $\rho$ and $\sigma$'s purifications respectively, and
$$F(\rho^{RQ},\sigma^{RQ})=F(\psi^{RQT}_{\rho},\psi^{RQT}_{\sigma})=|\langle\psi^{RQT}_{\rho}|{\psi^{RQT}_{\sigma}}\rangle|.$$
Now we can have the following
\begin{align*}
&F((\cI^R\otimes\mathcal{E})(\rho^{RQ}),(\cI^R\otimes\mathcal{F})(\sigma^{RQ}))\\
\geq&F((\cI^{RT}\otimes\mathcal{E})(\psi^{RQT}_{\rho}),(\cI^{RT}\otimes\mathcal{F})(\psi^{RQT}_{\sigma}))\\
=&F[\sum_{i=1}^m (I^{RT}\otimes E_i) \psi^{RQT}_{\rho}(I^{RT}\otimes E_i)^{\dag}, \\
&\ \sum_{i=1}^m (I^{RT}\otimes F_i)\psi^{RQT}_{\sigma} (I^{RT}\otimes F_i)^{\dag}]\\
\geq&\sum_{i=1}^m F[ (I^{RT}\otimes E_i) \psi^{RQT}_{\rho}(I^{RT}\otimes E_i)^{\dag},\\
&(I^{RT}\otimes F_i)\psi^{RQT}_{\sigma} (I^{RT}\otimes F_i)^{\dag}]\\
=&\sum_{i=1}^m |\langle \psi^{RQT}_{\rho}| (I^{RT}\otimes E_i)^{\dag}(I^{RT}\otimes F_i) |\psi^{RQT}_{\sigma}\rangle|\\
\geq & \frac{1}{\chi}\sum_{i=1}^m |\langle\psi^{RQT} _{\rho}|\chi_i(I^{RT}\otimes E_i^{\dag}F_i) |\psi^{RQT}_{\sigma}\rangle|\\
\geq & \frac{1}{\chi} |\langle \psi^{RQT}_{\rho}|(I^{RT}\otimes(\sum_{i=1}^m\chi_i E_i^{\dag}F_i) |\psi^{RQT}_{\sigma}\rangle|\\
=& \frac{1}{\chi} |\langle \psi^{RQT}_{\rho}|I^{RQT} |\psi^{RQT}_{\sigma}\rangle|\\
=& \frac{1}{\chi} |\langle \psi^{RQT}_{\rho}|\psi^{RQT}_{\sigma}\rangle|\\
=& \frac{1}{\chi} F(\psi^{RQT}_{\rho},\psi^{RQT}_{\sigma})\\
=& \frac{1}{\chi} F(\rho^{RQ},\sigma^{RQ}).
\end{align*}
The first inequality is due to Fact \ref{fac:monotonequantumoperation}, the monotonicity of the fidelity under partial trace. The second inequality is due to Fact \ref{strong concavity}, the strong concavity of the fidelity, and positive homogeneity.

Therefore, we choose $\zeta=\frac{1}{\chi}$.
\end{proof}

\section{main results}

Our main result is as follows
\begin{theorem} \label{main}
The Chernoff exponent for quantum operations, Eq. \ref{defineChernoffbound}, is finite if and only if they cannot be distinguished perfectly.
\end{theorem}

For two distinct quantum operations, $\mathcal{E}$ and $\mathcal{F}$, it is
straightforward to verify that $P_{err} \le \exp(-n\xi^\prime)$ for some $\xi^\prime>0$ by observing the following process. First, one can always find an input state $\rho$ such that $\mathcal{E}(\rho)$ and $\mathcal{F}(\rho)$ are distinct. Then we feed $\rho$ as input through the device for $n$ times. After that, the problem becomes to distinguish $\mathcal{E}(\rho)^{\otimes n}$ and $\mathcal{F}(\rho)^{\otimes n}$. Invoking the celebrated result on the Chernoff exponent for quantum states, we know that the error probability of distinguishing two different quantum states with identical copies decays according to an exponential function. Notice that this protocol only provides an upper bound on the minimal error probability of distinguishing $\mathcal{E}$ and $\mathcal{F}$, so one can conclude that $P_{err} \le \exp(-n\xi^\prime)$ for some $\xi^\prime>0$.

The above arguments show that the error decays at least exponentially. In other words, $\xi_{\cE,\cF}$ is greater than $0$. However, this scheme can be far from optimal. Perfect discrimination between unknown processes chosen from a finite set is shown to be possible. For two quantum operations that can be distinguished perfectly, $\xi_{\cE,\cF}=\infty$, we prove that this is the only case where $\xi_{\cE,\cF}=\infty$. Moreover, we provide an easy computable upper bound of $\xi_{\cE,\cF}$ for quantum operations that can not be distinguished perfectly, i.e., $P_{err} \ge \exp(-n\xi)$, where the parameter $\xi$ is a positive constant that depends on the two operations only.

\begin{proof} The only if part of Theorem \ref{main} is trivial. The if part follows Proposition 1 in Section II, and we prove it for prior probability distribution $\Pi_0=\Pi_1=1/2$. Also, we prove the Chernoff exponent is independent of a prior distribution in Proposition 2. 

To prove the only if part of Theorem \ref{main} under distribution $\Pi_0=\Pi_1=1/2$, we only need to show that when either condition in Proposition 1 is violated, the error probability is at least an exponential function of the number of channel uses.

First, we suppose $\mathcal{E}$ and $\mathcal{F}$ are joint, in the sense that the produced quantum states have non-zero overlapping supports for any common input
state, we show that there exists $\eta>0$ such that $P_{err,n}\geq {\eta^n}/{2}$ in the following:

Refer to Figure 1 for our notations.
By employing Lemma \ref{jointbound}, we observe that there exists $0\leq A_1\leq \rho_1,\sigma_1$ such that $\tr A_1\geq \eta$. Then, $0\leq A_1'=\mathcal{G}_1(A_1)\leq \rho_1',\sigma_1'$ such that $\tr A_1'=\tr A_1\geq \eta$. Then, there exists $0\leq A_2\leq \rho_2,\sigma_2$ such that $\tr A_2\geq \eta^2$. Thus, $0\leq A_2'=\mathcal{G}_2(A_2)\leq \rho_1',\sigma_1'$ such that $\tr A_2'=\tr A_2\geq \eta^2$ $\cdots$
There exists $0\leq A_n\leq \rho_n,\sigma_n$ such that $\tr A_n\geq \eta^n$.

By Helstrom's celebrated result on state discrimination \cite{Helstrom1967}, we know that the discrimination error satisfies the following
\begin{align*}
P_{err,min,n}&=\inf_{\rho_n,\sigma_n}\frac{1}{2}(1-\frac{\tr|\rho_n-\sigma_n|}{2})\\
&=\inf_{\rho_n,\sigma_n}\frac{1}{2}(1-\frac{\tr|\rho_n-A_n-\sigma_n+A_n|}{2})\\
&\geq \inf_{\rho_n,\sigma_n}\frac{1}{2}(1-\frac{\tr(\rho_n-A_n)+\tr(\sigma_n-A_n)}{2})\\
&= \frac{\tr A_n}{2}\\
&\geq \frac{\eta^n}{2}.
\end{align*}
The first inequality is according to the triangle inequality and $0\leq A_n\leq \rho_n,\sigma_n$.

Second, suppose two quantum
operations $\mathcal{E}$ and $\mathcal{F}$ can not transform non-orthogonal states into orthogonal
states. This is equivalent to $I\in\mathrm{span}\{E_i^*F_j\}$, as illustrated in Remark 1 at the end of Section II. We show in the following that there exists $\mu>0$ such that $P_{err,n}\geq \mu^n/4$.
The proof of this part is according to the observation that if two quantum operations cannot make nonorthogonal states orthogonal, they cannot change their fidelity significantly.

Refer to Figure 1 for our notations.
By employing Lemma \ref{fidelity}, we observe that there exists $\zeta>0$ such that after $n$ uses of the unknown quantum operation, the possible outcome states $\rho_n$ and $\sigma_n$ satisfy the following:
\begin{align*}
F(\rho_n,\sigma_n)&\geq \zeta F(\rho_{n-1}',\sigma_{n-1}')\\
                  &\geq \zeta F(\rho_{n-1},\sigma_{n-1})\\
                  &\geq \zeta^2 F(\rho_{n-2}',\sigma_{n-2}')\\
                  &\cdots\\
                  &\geq \zeta^{n-1} F(\rho_{1},\sigma_{1})\\
                   &\geq \zeta^{n} F(\psi,\psi)\\
                   &=\zeta^{n},
\end{align*}
where $F(\cdot,\cdot)$ denotes the fidelity of quantum states.

The first inequality is due to Lemma \ref{fidelity}.
The second inequality is due to the monotonicity of fidelity under any quantum operation.

According to the relation between fidelity and trace distance, we have
\begin{eqnarray*}
P_{err,min,n}&=&\inf_{\rho_n,\sigma_n}\frac{1}{2}(1-{\tr|\rho_n-\sigma_n|/2})\\
&\geq &\frac{1}{2}(1-\sqrt{1-\zeta^{2n}})\\
&\geq &\frac{\zeta^{2n}}{4},
\end{eqnarray*}
where the minimization ranges across all possible output $\rho_n$ and $\sigma_n$.

Therefore, we can choose $\mu=\zeta^2$.

Putting these two conditions together, we obtain that for indistinguishable quantum operations $\cE,\cF$ under uniform distribution,
\begin{equation}
\xi_{\cE,\cF}\leq \min\{-\log \eta, -\log \mu\}.
\end{equation}

\end{proof}

\iffalse
The following example shows that in some cases, our bound is tight.

\begin{Exam}
Suppose $\cE$ and $\cF$ are two constant quantum operations mapping 2-dimensional state Hilbert space to 3-dimensional state Hilbert space:
\begin{equation*}
\cE(\cdot) = \rho = \left[
\begin{array}{ccc}
p & & \\& 1-p & \\ & & 0\end{array}\right], \cF(\cdot) = \sigma = \left[\begin{array}{ccc}p & & \\ & 0 & \\ & & 1-p \end{array}\right]
\end{equation*}
where $(\cdot)$ represents any 2-dimensional input density operator, and $p$ is a parameter with $0<p<1$.
It is not difficult to check that $\eta = p$, $\zeta = p$ and $\mu = \zeta^2 = p^2$ satisfy the corresponding propositions or lemma, so we can bound the Chernoff exponent for $\cE$ and $\cF$:
$$\xi_{\cE,\cF} \leq \min\{-\log \eta, -\log \mu\} = -\log p.$$
On the other hand, the Chernoff exponent $\eta_{QCB}$ for $\rho$ and $\sigma$
$$\eta_{QCB} = -\log\left( \min_{0\le s\le1}\mathrm{Tr}(\rho^s\sigma^{1-s}) \right) = -\log p$$
is achievable if we use the device parallel, which implies that $\xi_{\cE,\cF}\ge\eta_{QCB}$. Therefore $\xi_{\cE,\cF}=-\log p$, which is exact the bound we are given.
\end{Exam}
\fi

 If we have more than two quantum operations, suppose we have a quantum device that is secretly chosen from $\{\cE_1,\cdots \cE_r\}$, a known set of quantum operations according to prior probability distribution $\{\Pi_1,\cdots,\Pi_r\}$. Our goal is to see which quantum operation the quantum device implements by using the device many times. The definition of the Chernoff exponent for two quantum operations Eq.(\ref{defineChernoffbound}) can be easily generalized into multiple quantum operations, where $P_{err,min,n}$ now is defined as the infimum error probability for distinguishing multiple quantum operations with $n$ uses. One can prove that the Chernoff exponent for multiple quantum operations shares the same properties as the Chernoff exponent for two quantum operations. This exponent does not depend on the prior probability distribution, and it is infinite if these quantum operations are mutually perfectly distinguishable. Moreover, it is, at most, the minimal mutual Chernoff exponent for quantum operations. 

\begin{Prop}
 The Chernoff exponent for multiple quantum operations, Eq. \ref{defineChernoffbound}, does not depend on the prior distribution. Moreover, the multi-channel Chernoff exponent is upper bounded by the smallest pairwise Chernoff exponent.
\end{Prop}

\begin{proof}

For prior $(\Pi_1,\Pi_2,\cdots,\Pi_r)$, $\Pi_1\geq\Pi_2\geq\cdots\Pi_r>0$ and fixed $n$, we use
$P_{err,min,n,\Pi}$ to denote the infimum error probability. $P_{err,min,n}$ denotes the infimum error probability for uniform prior $(1/r,1/r,\cdots,1/r)$.
For any discrimination scheme, we let $1-p_{n,i}$ be the probability of correctly identifying the $i$-th channel. Then we have
\begin{align*}
 \Pi_1\sum_{i=1}^r p_{n,i}\geq\sum_{i}\Pi_i p_{n,i}\geq \Pi_r\sum_{i=1}^r p_{n,i} \geq  \frac{\Pi_r}{2}(p_{n,k}+p_{n,l})
\end{align*}
for any $1\leq k,l\leq r$.

Since this inequality holds for the error probabilities in any strategy, one can just take the infimum over all strategies in each term of the inequality and have
\iffalse
\begin{align*}
r\Pi_1 (P_{err,min,n})&\geq r\Pi_1P=\Pi_1\sum_{i=1}^r p_{n,i} \geq  \sum_{i}\Pi_i p_{n,i}\geq 
P_{err,min,n,\Pi}\\
r\Pi_r (P_{err,min,n})&\geq r\Pi_rP\geq \Pi_r P_{err,min,n,\{k,l\}},
\end{align*}
where $P_{err,min,n,\{k,l\}}$ denotes the infimum error probability for distinguishing $\cE_{k}$ and $\cE_{l}$ for prior $(1/2,1/2)$.

Moreover, we have
\begin{align*}
P_{err,min,n,\Pi}\geq Q=\sum_{i}\Pi_i p_{n,i}\geq  \Pi_r\sum_{i=1}^r p_{n,i} \geq
r\Pi_r P_{err,min,n}.
\end{align*}
That implies
\begin{align*}
P_{err,min,n,\Pi}\geq r\Pi_r P_{err,min,n}.
\end{align*}
Therefore, 
\fi
\begin{align*}
r\Pi_1P_{err,min,n}\geq&  P_{err,min,n,\Pi}\\
\geq& r\Pi_rP_{err,min,n}\\
 \geq&  \Pi_r P_{err,min,n,\{k,l\}}
\end{align*}
Therefore,
\begin{align*}
-\varlimsup_{n\to \infty} \frac{P_{err,min,n}}{n}=&-\varlimsup_{n\to \infty} \frac{P_{err,min,n,\Pi}}{n} \\\leq& -\varlimsup_{n\to \infty} \frac{P_{err,min,n,\{k,l\}}}{n}.
\end{align*}
That is, the Chernoff exponent for multiple quantum operations does not depend on prior.
\end{proof}
According to the proof, we can also conclude that
\begin{Cor}
The Chernoff exponent for multiple quantum operations, Eq. \ref{defineChernoffbound}, is infinite if and only if the quantum operations are mutually perfectly distinguishable.
\end{Cor}
\begin{proof}
Suppose we are given a quantum operation $\cE$ being one of the quantum operations $\cE_1,\cE_2,\cdots,\cE_r$, and any two quantum operations can be distinguished perfectly. Let a protocol produce orthogonal quantum states $\rho_1$ and $\rho_2$ for quantum operations $\cE_1$ and $\cE_2$, respectively.

Now we run the protocol on $\cE$ and measure the output. We employ the measurement which can distinguish $\rho_1$ and $\rho_2$ perfectly. If the measurement outcome corresponds to $\rho_1$, then we know $\cE$ can not be $\cE_2$; otherwise, it can not be $\cE_1$.

Therefore, via finite uses of $\cE$, we can eliminate one candidate. By repeating this procedure, we can conclude that the Chernoff exponent is $\infty$.

Otherwise, if two quantum operations, say $\cE_1$ and $\cE_2$, can not be distinguished perfectly. According to the proof of Proposition 2, the multi-channel Chernoff exponent is upper bounded by the smallest pairwise exponent which is no more than the Chernoff exponent of $\cE_1$ and $\cE_2$, a finite number by Theorem 1.
\end{proof}
\section{Conclusion and Open Problems}

In this paper, we introduce the Chernoff exponent for quantum operations. We show the Chernoff exponent is finite if and only if the operations are not perfectly distinguishable. More precisely, we provide computable upper bounds of the Chernoff exponent by proving lower bounds on the error probability of distinguishing quantum operations with $n$ uses.
Our result is an asymptotic generalization of the diamond norm.

There are several open questions. One relates to the local operations and classical communication (LOCC)-Chernoff distance.
Motivated by the quantum Chernoff theorem \cite{ACMMABV,NussbaumSzkola06}, the LOCC-Chernoff exponent studies the distinguishability of two bipartite
mixed states under the constraint of LOCC, in the limit of many copies \cite{CVMB2010, MW2009}. There is a significant difference between the LOCC Chernoff exponent and the standard Chernoff exponent. Orthogonality does not indicate perfect LOCC distinguishability. More precisely, there exist quantum states which cannot be locally distinguished but multicopy makes them perfectly distinguishable \cite{YDY12, YDY2014}. This behavior is similar to the discrimination of quantum operations.
A fundamental question regarding the LOCC Chernoff exponent is still not answered: For two quantum states that are not LOCC perfectly distinguishable, even in the limit of many copies, does the LOCC discrimination error always decay exponentially? The first difficulty is we do not have a characterization of LOCC distinguishability of quantum states, even though this problem has been studied for more than 20 years \cite{YDY11, WSHV00, GKRS+01, BDMS+99, WAT05, HMM+06, BCJ+2015}.

We thank the editor and the anonymous reviewers whose comments have greatly improved this manuscript.
This work is supported by ARC Discovery Early Career Researcher Award DE180100156 and ARC Discovery Program DP210102449.
%---------References---------%
% \bibliographystyle{alpha}
% \bibliography{depolarizing}

\begin{thebibliography}{CEM{\etalchar{+}}15}

\bibitem{Acin2001}
A. Acin,
\newblock "Statistical distinguishability between unitary operations."
\newblock {\em Physical Review Letters}, \textbf{87}(17): 177901, 2001.

\bibitem{Aharonov1997}
D. Aharonov, A. Kitaev, and N. Nisan,
\newblock "Quantum circuits with mixed states."
\newblock {\em Proceeding of the
Thirtieth Annual ACM Symposium on Theory of Computation}, pp. 20-30, 1997.

\bibitem{ACMMABV}
K. M. R. Audenaert, J. Casamiglia, R. Munoz-Tapia, E. Bagan,
Ll. Masanes, A. Acin, and F. Verstraete,
\newblock "Discriminating states: the quantum Chernoff bound."
\newblock {\em Physical Review Letters}, \textbf{98}(16): 160501, 2007.

\bibitem{BCJ+2015}
S. Bandyopadhyay, A. Cosentino, N. Johnston, V. Russo, J. Watrous, and N. Yu,
\newblock "Limitations on separable measurements by convex optimization."
\newblock {\em IEEE Transactions on Information Theory}, \textbf{61}(6): 3593, 2015.

\bibitem{BDMS+99}
C. H. Bennett, D. P. DiVincenzo, T. Mor, P. W. Shor, J. A. Smolin, and B. M. Terhal,
\newblock "Unextendible product bases and bound entanglement."
\newblock{\em Physical Review Letters}, \textbf{82}(26): 5385, 1999.

\bibitem{Chernoff52}
H. Chernoff,
\newblock "A measure of asymptotic efficiency for tests of a hypothesis based on the sum of observations."
\newblock {\em The Annals of Mathematical Statistics}, \textbf{23}(4): 493, 1952.

\bibitem{CMW2016}
Tom Cooney, Milan Mosonyi, and Mark M. Wilde,
\newblock "Strong converse exponents for a quantum channel discrimination problem and quantum-feedback-assisted communication."
\newblock {\em Communications in Mathematical Physics}, \textbf{344}(3):  797-829, 2016.

\bibitem{CVMB2010}
J. Calsamiglia, J. I. de Vicente, R. Mu\~noz-Tapia, and E. Bagan,
\newblock "Local discrimination of mixed states."
\newblock {\em Physical Review Letters}, \textbf{105}(8): 080504, 2010.

\bibitem{Duan2007}
R. Duan, Y. Feng, and M. Ying,
\newblock "Entanglement is not necessary for perfect discrimination between unitary operations."
\newblock {\em Physical Review Letters}, \textbf{98}(10): 100503, 2007.

\bibitem{Duan2009}
R. Duan, Y. Feng, and M. Ying,
\newblock "Perfect distinguishability of quantum operations."
\newblock {\em Physical Review Letters}, \textbf{103}(21): 210501, 2009.

\bibitem{FuchsG99}
C. A. Fuchs, J. Van De Graaf,
\newblock "Cryptographic distinguishability measures for quantum-mechanical states."
\newblock {\em IEEE Transactions on Information Theory}, \textbf{45}(4): 1216, 1999.
  1999.

\bibitem{GKRS+01}
S. Ghosh, G. Kar, A. Roy, A. Sen(De) and U. Sen,
\newblock "Distinguishability of Bell states."
\newblock {\em Physical Review Letters}, \textbf{87}(27): 277902, 2001.

\bibitem{Hayashi09}
M. Hayashi,
\newblock "Discrimination of two channels by adaptive methods and its application to quantum system."
\newblock {\em IEEE Transactions on Information Theory}, \textbf{55}(8): 3807, 2009.

\bibitem{Helstrom1967}
C. W. Helstrom,
\newblock "Detection theory and quantum mechanics."
\newblock {\em Information and Control}, \textbf{10}(3): 254, 1967.

\bibitem{HHDW2010}
A. W.  Harrow, A. Hassidim, D. W. Leung, and J. Watrous
\newblock "Adaptive versus non-adaptive strategies for quantum channel discrimination."
\newblock {\em Physical Review A}, \textbf{81}(3): 032339, 2010.

\bibitem{HMM+06}
M. Hayashi, D. Markham, M. Murao, M. Owari and S. Virmani,
\newblock "Bounds on multipartite entangled orthogonal state discrimination using local operations and classical communication."
\newblock {\em Physical Review Letters}, \textbf{96}(4): 040501, 2006.

\bibitem{Holevo1972}
A. S. Holevo,
\newblock "An analog of the theory of statistical decisions in noncommutative probability theory."
\newblock {\em Transactions of the Moscow Mathematical Society}, \textbf{26}: 133, 1972.

\bibitem{Ji2006}
Z. Ji, Y. Feng, R. Duan, and M. Ying,
\newblock "Identification and distance measures of measurement apparatus."
\newblock {\em Physical Review Letters}, \textbf{96}(20): 200401, 2006.

\bibitem{Kitaev1997}
A. Kitaev,
\newblock "Quantum computations: Algorithms and error correction."
\newblock {\em Russian Mathematical Surveys}, \textbf{52}(6): 1191, 1997.

\bibitem{Laing2009}
A. Laing, T. Rudolph, and J. L. O'Brien,
\newblock "Experimental quantum process discrimination."
\newblock {\em Physical Review Letters}, \textbf{102}(16): 160502, 2009.


\bibitem{LQ2008}
L. Li and D. Qiu,
\newblock "Local entanglement is not necessary for perfect discrimination between unitary operations acting on two qudits by local operations and classical communication."
\newblock {\em Physical Review A}, \textbf{77}(03): 032337 , 2008.


\bibitem{Li2016}
K. Li,
\newblock "Discriminating quantum states: the multiple Chernoff distance."
\newblock {\em Annals of Statistics}, 44 (4): 1661-1679 2016.

\bibitem{MW2009}
W. Matthews, A. Winter,
\newblock "On the Chernoff distance for asymptotic LOCC discrimination of bipartite quantum states."
\newblock {\em Communications in Mathematical Physics}, \textbf{285}: 161, 2009.

\bibitem{NielsenC00}
M. A. Nielsen, I. Chuang,
\newblock "Quantum computation and quantum information."
\newblock {\em Cambridge University Press}, Cambridge, UK, 2000.

\bibitem {NussbaumSzkola06}
M. Nussbaum, A. Szko{\l}a,
\newblock "The Chernoff lower bound for symmetric quantum hypothesis testing."
\newblock {\em The Annals of Statistics}, \textbf{37}(2): 1040, 2009.



\bibitem{Parth01}
K. R. Parthasarathy,
\newblock "On consistency of the maximum likelihood method in testing multiple quantum hypotheses."
\newblock {\em Stochastics in Finite and Infinite Dimensions}, Birkh\"auser Boston, 361, 2001.


\bibitem{Puzzuoli2016}
D. Puzzuoli, J. Watrous,
\newblock "Ancilla dimension in quantum operation discrimination."
\newblock {\em Annales Henri Poincar\'{e}}, Springer International Publishing, 2016.


\bibitem{TW2016}
Masahiro Takeoka, and Mark M. Wilde,
\newblock "Optimal estimation and discrimination of excess noise in thermal and amplifier channel."
\newblock {\em arXiv:1611.09165}, 2016.

\bibitem{Uhl77}
A. Uhlmann,
\newblock "The ``transition probability" in the state space of a *-algebra",
\newblock {\em Reports on Mathematical Physics}, \textbf{9}(2): 273, 1976.

\bibitem{WAT05}
J. Watrous,
\newblock "Bipartite subspaces having no bases distinguishable by local operations and classical communication."
\newblock {\em Physical Review Letters}, \textbf{95}(8): 080505, 2005.

\bibitem{Watrous2009}
J. Watrous,
\newblock "Semidefinite programs for completely bounded norms."
\newblock {\em Theory of Computing}, \textbf{5}(11): 217, 2009.

\bibitem{Watrous2015}
J. Watrous,
\newblock "Theory of quantum information."
\newblock {\em University of Waterloo Fall}, 128, 2011.

\bibitem{WSHV00}
J. Walgate, A. J. Short, L. Hardy and V. Vedral,
\newblock "Local distinguishability of multipartite orthogonal quantum states."
\newblock {\em Physical Review Letters}, \textbf{85}(23): 4972, 2000.

\bibitem{YDY11}
N. Yu, R. Duan and M. Ying,
\newblock "Any $2\otimes n$ subspace is locally distinguishable."
\newblock {\em Physical Review A}, \textbf{84}(1): 012304, 2011.

\bibitem{YDY12}
N. Yu, R. Duan and M. Ying,
\newblock "Four locally indistinguishable ququad-ququad orthogonal maximally entangled states."
\newblock {\em Physical Review Letters}, \textbf{109}(2): 020506, 2012.

\bibitem{YDY2014}
N. Yu, R. Duan, and M. Ying,
\newblock "Distinguishability of quantum states by positive operator-valued measures with positive partial transpose."
\newblock {\em IEEE Transactions on Information Theory}, \textbf{60}(4): 2069, 2014.


\bibitem{ZZG2007}
X. F. Zhou, Y. S. Zhang, and G. C. Guo
\newblock "Unitary Transformations Can Be Distinguished Locally."
\newblock {\em Physical Review Letters}, \textbf{99}(17): 170401, 2007.


\bibitem{Slater1950}
M. Slater, 
\newblock "Lagrange Multipliers Revisited."
\newblock Cowles Commission Discussion Paper No. 403 (Report).











\end{thebibliography}

\newcommand{\etalchar}[1]{$^{#1}$}

\begin{IEEEbiographynophoto}{Nengkun Yu}
is a Senior Lecturer in the Centre for Quantum Software and Information, University of Technology Sydney. He received the B.S. and Ph.D. degrees from the Department of Computer Science and Technology, Tsinghua University, Beijing, China, in July of 2008 and 2013. From January 2014 to July 2016, Nengkun was a postdoc at the Institute for Quantum Computing at the University of Waterloo, Canada. His research focuses on quantum computing. 
\end{IEEEbiographynophoto}

\begin{IEEEbiographynophoto}{Li Zhou}
is a postdoc at Max Planck Institute for Security and Privacy. He received my Ph.D. in Computer Science and Technology from Tsinghua University in 2019. His research focuses on quantum programs and protocols, including static analysis and verification, and runtime debugging and testing. 
\end{IEEEbiographynophoto}

\end{document}